\documentclass[11pt,a4paper]{article}
\usepackage[T1]{fontenc}
\usepackage{lmodern}
\usepackage[margin=24mm]{geometry}
\usepackage{amsmath,amssymb,amsthm,mathtools}
\usepackage{microtype,booktabs,graphicx}
\usepackage[colorlinks=true,allcolors=blue]{hyperref}
\newtheorem{theorem}{Theorem}
\newtheorem{lemma}[theorem]{Lemma}
\newtheorem{proposition}[theorem]{Proposition}

\DeclareMathOperator{\Rep}{Rep}
\DeclareMathOperator{\Tr}{Tr}

\DeclareMathOperator{\End}{End}
\DeclareMathOperator{\FPdim}{FPdim}
\newcommand{\HH}{\mathcal H}
\newcommand{\one}{\mathbf1}
\newcommand{\CC}{\mathbb C}
\newcommand{\ZZ}{\mathbb Z}
\newcommand{\QQ}{\mathbb Q}
\newcommand{\PF}{\mathrm{PF}}
\title{Haagerup conformal field theory at central charge eight}
\author{Tzu-Chen Huang}
\date{\today}
\hypersetup{pdftitle={Haagerup conformal field theory at central charge eight},pdfauthor={Tzu-Chen Huang}}
\begin{document}
\maketitle
\begin{abstract}
We construct a chiral algebra $U$ of central charge eight inside the
lattice vertex operator algebra $V_{E_6\oplus A_2}$. It is generated by
two commuting affine $\mathfrak{sl}_2$ current algebras at levels $1$ and $39$,
the stress tensor of the lattice theory, and a single weight-two field
built from three lattice exponentials. We announce that $U$ is strongly
rational and completely unitary, and that $\Rep(U)$ is equivalent, as a
unitary ribbon category, to the Drinfeld center $Z(\HH_3)$ of the
Haagerup fusion category $\HH_3$. Its characters are the $\gamma=0$
Evans--Gannon vector. From $U$ we obtain a diagonal full CFT, and also a
holomorphic $E_8$ extension on which $\HH_3$ acts by topological defects.
\end{abstract}

\section{Introduction}
Fusion categories describe symmetries of two-dimensional conformal
field theories whose topological lines need not be invertible. Many
familiar fusion categories arise from finite groups, quantum groups or
affine Lie algebras. The Haagerup categories are among the simplest
examples outside these standard constructions, so it is natural to ask
whether some CFT has Haagerup symmetry. Numerical studies of Haagerup anyon chains and
lattice models found evidence for such a CFT at $c\approx2$
\cite{HLOTT,VLVWOHV}. Here we give an explicit construction of a Haagerup CFT at $c=8$,
starting from its chiral algebra.

Write $\HH_3$ for the unitary category with simples $\one,\alpha,\alpha^2,\rho,\alpha\rho,\alpha^2\rho$ and fusion rules
\begin{equation}\label{eq:HI}
 \alpha^3=\one,\qquad \rho\alpha=\alpha^2\rho,\qquad
 \rho^2=\one+\rho+\alpha\rho+\alpha^2\rho,
 \qquad d_\rho=\frac{3+\sqrt{13}}2.
\end{equation}
The fusion ring alone does not determine the category: $\HH_2$ has the
same fusion ring but a different associator, and the two are Morita
equivalent \cite{Izumi,GS}. Evans and Gannon studied the modular data of
the center and worked out the possible character vectors of a Haagerup
VOA at $c=8$ \cite{EG}. M\"oller and Rayhaun
\cite{MR} observed that realizing $Z(\HH_3)$
as the representation category of a positive, strongly rational $c=8$
VOA is equivalent to finding an $\HH_3$ symmetry of the chiral
$(E_8)_1$ theory. This relation was discussed further by Gannon and
Rayhaun \cite{GR}. Albert, Honda, Kaidi and Zheng
\cite{AHKZ} proposed that gauging the diagonal Haagerup symmetry of
the non-chiral $(E_8)_1$ model gives the $c=8$ theory predicted by
Evans and Gannon.
Our approach is to write down a subalgebra of a lattice vertex operator algebra (VOA), so
that every product we use is inherited from a fixed, well-understood
theory. Our main result is the following.

\begin{theorem}\label{thm:main}
Let $U$ be the algebra defined in \eqref{eq:defU}. With the finite
exact lattice inputs specified in Sections~\ref{sec:closure}--\ref{sec:branching},
$U$ is strongly rational and completely unitary, with
\begin{equation}\label{eq:main}
 c=8,\qquad U_1\cong\mathfrak{sl}_2\oplus\mathfrak{sl}_2,
 \qquad \Rep(U)\simeq Z(\HH_3)
\end{equation}
as unitary ribbon categories. Its twelve irreducible characters are
the $\gamma=0$ Evans--Gannon vector. There is a diagonal unitary full
CFT with $c_L=c_R=8$. There is also a unitary holomorphic extension
$H'\supset U$, isomorphic to $V_{E_8}$, whose extension algebra
$A_3\in\Rep(U)$ has all-module category $\Rep(U)_{A_3}\simeq\HH_3$
and whose algebra-bimodule defects form
$\HH_3\boxtimes\HH_3^{\mathrm{op}}$.
\end{theorem}

By strong rationality we mean that $U$ is simple, of CFT type,
self-contragredient, rational and $C_2$-cofinite. Complete unitarity
means that every simple module carries a positive invariant form and
that the canonical fusion forms are positive. In the sketches below we
keep apart two kinds of input: a finite number of identities checked
by computer algebra, and the infinite-dimensional statements we deduce
from them. The computer-algebra inputs are exact identities in lattice
coordinates. Matching finitely many character coefficients would not be
enough for our purposes. In this note we give the defining fields of $U$ and an outline of
the argument identifying the resulting CFT. The full coefficient lists, and records that
allow the computations to be replayed, will be made available later
with a detailed version. 

The $E_8$ theory of Theorem~\ref{thm:main} can furthermore be compared with the
numerical modular bootstrap of \cite{AHKZ}, which bounds scalar
Virasoro-primary gaps in theories with $\HH_3$ symmetry. Let $\HH_3$ act through the first factor
in \eqref{eq:defects}, and write $\mathcal V_\mu$ for the tube-sector
multiplicity spaces as in \cite{AHKZ}. The twelve tube-sector partition
functions are $\chi_\mu\overline{\chi_{E_8}}$. The labels $(0,b,a,c_0)$
of Section~\ref{sec:characters} correspond to $(1,\pi_1,\pi_2,\sigma_1)$
in \cite{AHKZ}. Since $\chi_{E_8}$ has integral weights, scalars occur
only in these four sectors. In each, the lightest nonidentity scalar
Virasoro primary has $(h,\bar h)=(1,1)$. Hence the
$\HH_3$-invariant scalar gap and the scalar gap in
$\mathcal V_{\pi_2}\cup\mathcal V_{\sigma_1}$ are both $\Delta=2$, at
$c_L=c_R=8$. Figure~\ref{fig:bootstrap} places these points on the
published bounds \cite[Figs.~2 and~3]{AHKZ}; both lie below the bounds
and are therefore not excluded. This is only a consistency check: $(8,2)$ is not at a
kink of either curve, and we have not rerun the semidefinite programs.
\begin{figure}[htbp]
 \centering
 \includegraphics[width=\textwidth]{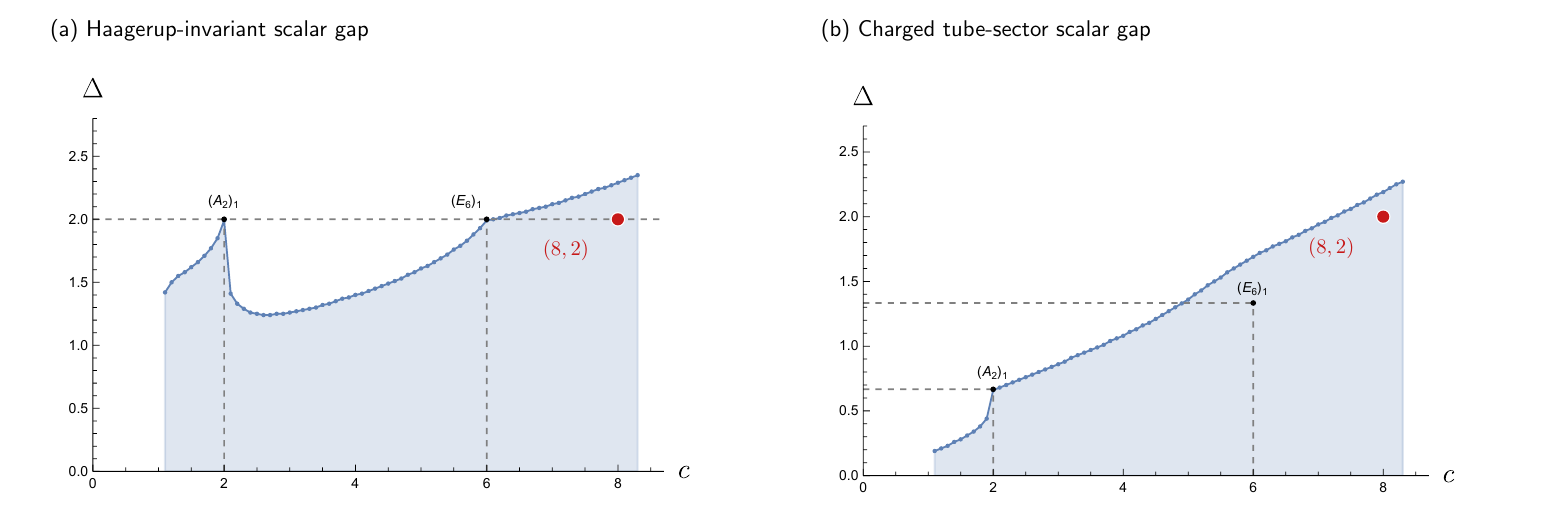}
 \caption{Scalar Virasoro-primary gaps for the first-factor $\HH_3$
 symmetry of the $c_L=c_R=8$ $E_8$ theory, shown in red.
 Original panels from J.~Albert, Y.~Honda, J.~Kaidi, and Y.~Zheng,
 \emph{Haagerup Symmetry in $(E_8)_1$?}, Phys.~Rev.~Lett.~136,
 091603 (2026), \href{https://doi.org/10.1103/6tzz-tvp7}{doi:10.1103/6tzz-tvp7},
 reproduced under CC~BY~4.0. Red markers and labels have been added.
 The left panel bounds $\mathcal V_1$; the right panel bounds
 $\mathcal V_{\pi_2}\cup\mathcal V_{\sigma_1}$.}
 \label{fig:bootstrap}
\end{figure}
\section*{Note added in version $2$}
Version 1 of this manuscript was submitted to arXiv on October 9, 2026, at 20:16 CST.
On October 11, 2026, the preprint \cite{Riesen} was announced, independently establishing
the construction of the Haagerup VOA.
\section{The explicit chiral algebra}\label{sec:definition}
Let $L=E_6\oplus A_2$, with simple-root basis $\alpha_0,\ldots,\alpha_7$.
The Gram matrix $G$ has diagonal entries $2$ and entries $-1$ on the
edges
\[
 (0,1),\ (1,2),\ (2,3),\ (3,4),\ (2,5),\ (6,7),
\]
and zero elsewhere. Set $B=V_L\subset H=V_{E_8}$, using the usual
index-three lattice extension, and fix
\begin{equation}\label{eq:cocycle}
 \epsilon(a,b)=(-1)^{\sum_i a_i b_i+\sum_{i>j}a_i b_jG_{ij}}.
\end{equation}
Coordinates of lattice momenta refer throughout to this basis. The
VOA convention is $Y(v,z)=\sum_n v_nz^{-n-1}$ and $T=L(-1)$; in
particular
\begin{equation}\label{eq:lattice}
 Y(e^a,z)e^b=\epsilon(a,b)z^{(a,b)}
 \exp\!\left(\sum_{n>0}\frac{a(-n)}n z^n\right)e^{a+b}.
\end{equation}
Together with the Heisenberg relations this fixes the products used below.

Put $\alpha_L=(-1,-2,-3,-2,-1,-2,0,0)$,
$c=(5,8,9,8,5,0,2,2)$ and $h_R=\sum_i c_i\alpha_i$. In a rational
torus gauge take the currents
\begin{align}
 (H_L,E_L,F_L)&=\Bigl(\alpha_L(-1)\one,e^{\alpha_L},-e^{-\alpha_L}\Bigr),\label{eq:left}\\
 (H_R,E_R,F_R)&=\Bigl(h_R(-1)\one,
       \sum\nolimits_{c_i\ne0}e^{\alpha_i},-\sum\nolimits_i c_i e^{-\alpha_i}\Bigr).
       \label{eq:right}
\end{align}

The two triples commute with each other. Their Cartan elements have
norms $2$ and $78$, and $(E_L)_1F_L=\one$, $(E_R)_1F_R=39\cdot\one$, so
together they generate the simple affine algebra
$\mathcal K=A_{1,1}\otimes A_{1,39}$, where
$A_{1,k}=L_k(\mathfrak{sl}_2)$. The level $39$ can be understood as
$35+4$, coming from the principal embeddings into $A_5$ and $A_2$.

Define three norm-four momenta and a field
\begin{gather}
 a=\alpha_6+\alpha_7-\alpha_5,\qquad
 b=\alpha_3+\alpha_4-\alpha_5,\qquad
 c'=\alpha_0+\alpha_1-\alpha_5,\notag\\
 u=e^a+\frac{\sqrt{13}+3\sqrt{-3}}{20}e^b
          +\frac{\sqrt{13}-3\sqrt{-3}}{20}e^{c'}.
          \label{eq:u}
\end{gather}

This $u$ is a highest-weight vector of conformal weight two, with Cartan
weights $(1,13)$; equivalently, it generates current multiplets of
dimensions $(2,14)$. We define

\begin{equation}\label{eq:defU}
 \boxed{\quad U=\langle\mathcal K,\omega_B,u\rangle\subset B,\qquad
 \omega_B=\frac12\sum_{i,j}(G^{-1})_{ij}
               \alpha_i(-1)\alpha_j(-1)\one.\quad}
\end{equation}
We use the stress tensor $\omega_B$ of the lattice theory, so $c=8$.
The Sugawara conformal vector of $\mathcal K$ has central charge
$1+117/41<8$ and is not the conformal vector used for $U$.

The rational gauge above is convenient for exact arithmetic, but the
positive form has to be transported to it from the standard unitary
gauge. Set
$(\lambda_0,\ldots,\lambda_7)=(5,8,9,8,5,1,2,2)$
and let $g$ multiply $e^a$ by
$\prod_i\lambda_i^{-a_i/2}$, fixing oscillators. Transport the
lattice Hermitian form and PCT by $g$. On real oscillator monomials,
\begin{equation}\label{eq:PCT}
 \theta_g\bigl(h_1(-n_1)\cdots h_r(-n_r)e^a\bigr)
 =(-1)^r\Bigl(\prod_i\lambda_i^{a_i}\Bigr)
       h_1(-n_1)\cdots h_r(-n_r)e^{-a},
\end{equation}
antilinearly on coefficients. The form is
$\langle v,w\rangle_g=\langle g^{-1}v,g^{-1}w\rangle$.
In the original unitary gauge, $g^{-1}u$ is proportional to
$e^a+\zeta e^b+\bar\zeta e^{c'}$, with
$\zeta=(\sqrt{13}+3i\sqrt3)/\sqrt{40}$ and $|\zeta|=1$.
So the coefficients of $u$ amount to a relative phase, and there is no
continuous parameter. The positive real character that defines $g$
extends to $E_8$, and we use this one unitary structure for the whole
chain $U\subset B\subset H$.
\section{Finite closure}\label{sec:closure}
The first thing to check is that $U$ is really smaller than $B$: a
priori, repeated OPEs of the generators could regenerate the other
eighty currents of $B$, and a computation in the first few weights
cannot exclude this. We therefore use a criterion that controls all
weights at once. 

Set $S=\bigoplus_d S_d$, where $S_d\subset B_d$.
The \emph{differential normal span} $\mathcal N(S)$ is the smallest
linear subspace of $B$ containing $S$ and the vacuum and closed
under translation and $(-1)$ products.

\begin{lemma}[Recursive strong generation]\label{lem:closure}
Suppose $h\ge2$ and $S_d$ are finite-dimensional for $1\le d\le h$. Write
$K=S_1$ and $A=S_2$. Assume all singular products
$k_ns,a_ns$, for $k\in K$, $a\in A$, $s\in S$ and $n\ge0$, belong
to $\mathcal N(S)$. Suppose further that, for $3\le d\le h$,
\begin{equation}\label{eq:recursive}
 S_d\subset A_0S_{d-1}+\mathcal N_d(S_1+\cdots+S_{d-1}).
\end{equation}
Then $\mathcal N(S)$ is a vertex subalgebra.
\end{lemma}
\begin{proof}[Proof sketch]
The zero modes $a_0$ are derivations commuting with $T$, so the first
hypothesis already implies $a_0\mathcal N(S)\subset\mathcal N(S)$.
Induct on the total weight of a pair of generators, and then on its
left weight. Replace a higher left generator by $a_0b$ plus a lower
normal expression and use
\begin{equation}\label{eq:Jacobi}
 (a_0b)_nt=a_0(b_nt)-b_n(a_0t).
\end{equation}
The first term has smaller input weight. In the second term every
composite or derivative in $a_0t$ reduces by the filtered Wick identity;
a single generator leaves the same total weight and a smaller left
weight. Normal remainders are handled by the same induction using
\[
 [x_\lambda:yz:]=:[x_\lambda y]z:+:y[x_\lambda z]:
                  +\int_0^\lambda[[x_\lambda y]_\mu z]\,d\mu.
\]
The integral has smaller total input weight. All generator singular
products therefore close, and Wick induction extends closure to the
entire normal span \cite{DSK}.
\end{proof}

We now describe the input to the lemma. The exact lattice calculation
gives a PCT-stable space
$S_2\subset U_2$ of dimension $120$, generated from the specified
currents, $\omega_B$ and $u$ by recorded mode recipes. Its current
multiplets are
\[
 (1,5),\ (1,9),\ (1,9),\ (1,13),\ (2,2),\ (2,6),\ (2,8),\ (2,14).
\]
They contribute $96$ states, in addition to $22$ affine vacuum
descendants and two neutral fields. The constructed space $S_3$ has
dimension $956$, with
\begin{equation}\label{eq:S3}
 N_3=K_{-1}S_2+TS_2,\qquad \dim N_3=634,\qquad
 S_3=N_3+(S_2)_0S_2.
\end{equation}
Applying current projectors to $(S_2)_0S_3$ gives a current triplet
$P_4\subset U_4$ of highest weight $(0,2)$. Set
\begin{align}
 N_4&=K_{-1}S_3+TS_3+(S_2)_{-1}S_2,&W_4&=N_4+P_4,\notag\\
 N_5&=K_{-1}W_4+TW_4+(S_2)_{-1}S_3.&&\label{eq:normal}
\end{align}
By construction $P_4\subset(S_2)_0S_3+N_4$.
Apart from the singular-product identities in lower weight, the two key checks are
\begin{equation}\label{eq:gates}
 (S_2)_0S_3\subset W_4,\qquad (S_2)_0P_4\subset N_5.
\end{equation}
We verify these as equalities of vectors over
$\QQ(\sqrt{13},\sqrt{-3})$, and we record how each field was generated.

To give a sense of the size of the computation: the quotient
$S_3/N_3$ has dimension $322$. For the first check, we verify the
inclusion on a set of seed vectors. Together with the two
commuting current raising operators, these span the quotient. If a quotient
$M$ of it satisfies $M=E_LM+E_RM$, then $M=0$ because the raising
operators are jointly nilpotent, so the identities propagate from the
seeds. The remaining classes that are singlets under both currents are
recovered using exact current projectors and lowering identities that
are checked separately. The second check uses the way the triplet was
constructed, the Jacobi identities, and an explicit witness for the
bottom current.

We now apply Lemma~\ref{lem:closure} to $K,S_2,P_3,P_4$, where $P_3$
is a space of lifts of $S_3/N_3$ to $U_3$. The outputs we need have
weight at most five. The subalgebra produced by the lemma contains the
defining fields of $U$ and is contained in the algebra they generate,
so it is $U$. Hence $U$ is strongly generated in weights at most four,
and
\begin{equation}\label{eq:lowgrades}
 \dim U_1=6,\qquad \dim U_2=120,\qquad \dim U_3=956.
\end{equation}
Since the generators are PCT stable, the vacuum form on $U$ is positive
and $U$ is simple. In particular $U\ne B$, because $\dim B_1=86$.

\section{From positive stress identity to rationality}\label{sec:rationality}

For rationality we need two more pieces of finite input. The first is
that $U$ contains copies of the parafermion algebras
$K(\mathfrak{sl}_2,3)$ and $K(\mathfrak{sl}_2,13)$, which are regular
and unitary \cite{ALY}. Their conformal vectors $p$ and $z$ have
central charges $4/5$ and $8/5$, and on their states their own grading
and translation agree with those inherited from $B$. The second is the
exact weight-two identity
\begin{equation}\label{eq:positive}
 \omega_U=\omega_{A_{1,1}}+\frac{3403}{4563}\omega_{A_{1,39}}
                      +\frac{25}{6}\bar p+\frac{25}{26}\bar z.
\end{equation}
A bar denotes averaging over the compact symmetry group generated
by the current zero modes. This projects each field onto its
invariant component. The group acts unitarily with respect to the inner product
$\langle\cdot,\cdot\rangle_g$ defined in
Section~\ref{sec:definition}. Note that
\eqref{eq:positive} is not a decomposition of $\omega_U$ into
commuting pieces. Each average lies in the convex hull of a compact
subset of a finite-dimensional space, so it can be replaced by a
finite convex combination. We can therefore write
$\omega_U=\sum_i t_ie_i$ with all $t_i>0$, where each $e_i$ is the
conformal vector of a regular, unitary subalgebra $A_i$ whose grading
is compatible with that of $U$.

Identifying the second parafermion algebra takes more than matching
its central charge. We check that its weight-three primary $w$, its
stress tensor $z$, and the fields $w_1w$ and $w_1(w_1w)$ have exactly
the singular OPEs of $K(\mathfrak{sl}_2,13)$ in the first row,
including the simple poles and the decoupling relation at weight six
\cite{DLY,Linshaw}. Iterating and using Wick identities give the
remaining OPEs among the generators. Now consider, in the product of
the two simple realizations, the subalgebra generated by the diagonal
pairs. It is of CFT type, so it has a unique maximal proper graded
ideal. The kernels of the two projections therefore coincide, and this
gives the identification. The finite inputs to this argument are
membership in $U$, PCT stability, the complete first rows, and
\eqref{eq:positive}.

\begin{proposition}\label{prop:rational}
The positive identity \eqref{eq:positive}, with these embedded
regular factors, implies that $U$ is strongly rational. Every
nonvacuum simple module has strictly positive lowest weight.
\end{proposition}
\begin{proof}[Proof sketch]
We begin with $C_2$-cofiniteness. For every simple $A_i$-module, the
quotient by its $C_2$ subspace is finite-dimensional \cite{ABD}. Since
there are only finitely many simple $A_i$-modules, the energies
occurring in these quotients are bounded by some $D_i$. By regularity,
$U$ decomposes as a direct sum of $A_i$-modules. If
$v\in U_n\cap C_2(U)^\perp$, then $v$ has no component of
$A_i$-energy above $D_i$, since such components already lie in
$C_2(U)$. As the $L^{e_i}(0)$ are self-adjoint, we get
\begin{equation}\label{eq:C2}
 n\|v\|^2=\sum_i t_i\langle L^{e_i}(0)v,v\rangle
       \le\Bigl(\sum_i t_iD_i\Bigr)\|v\|^2.
\end{equation}
So the $C_2$ quotient vanishes above a finite weight, and $U$ is
$C_2$-cofinite.

Next, let $M$ be a simple $U$-module with lowest weight $h$ and top
space $E$. Because the gradings are compatible, $E$ is a sum of
complete tops of $A_i$-modules. Each partial energy $L^{e_i}(0)$ is
diagonalizable on $E$ with nonnegative real eigenvalues, so
\begin{equation}\label{eq:toptrace}
 h\dim E=\sum_i t_i\Tr_E L^{e_i}(0)\ge0.
\end{equation}
If $h=0$, then every partial energy vanishes on $E$. Since the vacuum
is the only simple $A_i$-module of weight zero, each partial
translation annihilates $E$, and hence $L(-1)E=0$. Any nonzero $v\in E$
then gives a module map $U\to M$, $a\mapsto a_{-1}v$, and so
$M\cong U$.

Applying the same trace argument to the generalized weight-zero space
of an epimorphism $N\twoheadrightarrow U$, we find a vacuum-like
preimage of $\one$, which splits the epimorphism. So the vacuum module
is projective. The category of grading-restricted generalized $U$-modules
is a finite braided category with Grothendieck--Verdier duality, by
$C_2$-cofiniteness \cite{McRae}. For any indecomposable object, a
nonzero evaluation onto the unit splits because the unit is projective
and simple, so every indecomposable is non-negligible. A finite
category has at most exponential growth, so the rigidity theorem of
Etingof and Penneys for moderate growth applies \cite{EP}. In a rigid
category with projective unit every object is projective. Hence the
category is semisimple and $U$ is rational. Note that this argument
does not assume that the nonvacuum modules are unitary.
\end{proof}

\section{Branching determines the boundary fusion ring}\label{sec:branching}
To find how $B$ and $H$ decompose as $U$-modules, we need a uniform
bound on the lowest weights of $U$-modules. This uses one more piece of
finite input: a rational lattice--Potts subalgebra
$\mathcal T\subset U$ of central charge $14/5$, commuting with
$\PF_{13}=K(\mathfrak{sl}_2,13)$. The product
$A=\mathcal T\otimes\PF_{13}$ has pointed charges
$(j,k)\in(\ZZ/13)^2$ with quadratic form $-(j^2+k^2)/13$.
In $U_2$ there are joint eigenvectors at Cartan weights
$(\pm1,\pm5),(0,\pm8)$ with $\mathcal T$-energy $14/13$, parafermion
energy $12/13$, and zero complementary energy. They lie on the
simple-current line $k=-5j$. In the zero-complementary-energy subalgebra
relative to $A$, positivity, nonzero inverse-charge products and
associativity imply that the supported invertible constituents have multiplicity one and
subgroup support; integrality of the grading makes this subgroup
isotropic. Adjoining these vectors to $A$ therefore gives a regular
extension inside $U$,
\begin{equation}\label{eq:core}
 A_L=\bigoplus_{k=-5j}\mathcal T_{1,j}\otimes\PF_{0,k}\subset U,
 \qquad \Rep(A_L)\simeq\mathrm{Fib}\boxtimes\mathrm{SO}(3)_{13}.
\end{equation}
The extension is free on charges and has fourteen simple modules,
denoted $x_a=(1,a)$ and $y_a=(\tau,a)$, $0\le a\le6$.

For a $U$-top $E$, let $x_a,y_a$ also denote the multiplicities of
these complete tops, and put $D=\dim E$. If $L,Q,P,Z$ denote the
traces of the left affine energy, $H_R(0)^2$, Potts energy, and
parafermion energy, respectively, \eqref{eq:positive} gives
\begin{equation}\label{eq:W}
 hD=W=L+\frac{83}{6084}Q+\frac{25}{6}P+\frac{25}{26}Z.
\end{equation}
The dimensions and traces for the fourteen tops are obtained by a
finite minimization over charges. For $a=0,\ldots,6$ they are
\[
\begin{array}{c|rrrrrrr}
 a&0&1&2&3&4&5&6\\\hline
 \dim x_a&1&1&5&9&2&14&6\\
 4680W(x_a)&0&600&10440&56640&5520&148560&54240\\
 \dim y_a&1&7&15&2&14&6&2\\
 4680W(y_a)&7800&36960&110880&4080&81480&18240&4920
\end{array}
\]
Since the full $U$-top $E$ carries a complete representation of the right
$\mathfrak{sl}_2$, its multiplicities satisfy $y_2\ge x_5$ and
$y_4\ge x_5+x_6$. The following exact identity then bounds the trace:
\begin{align}\label{eq:topbound}
 4680(\tfrac{17}{10}D-W)
 ={}&8460(y_2-x_5)+28716(y_4-x_5-x_6)\notag\\
 &+7956x_0+7356x_1+29340x_2+14964x_3+10392x_4\notag\\
 &+22212x_6+156y_0+18732y_1+11832y_3\notag\\
 &+1188y_4+29496y_5+10992y_6\ge0.
\end{align}
So every simple $U$-module has lowest weight $0\le h\le17/10<2$. Note
that the estimate only uses traces on the top; we do not need the
partial energies to be simultaneously diagonalizable, nor a positive
form on $E$.

The last finite inputs are that the zero modes act irreducibly on the
eighty currents of $B$ outside $U$, and on each of the four mixed
lattice tops of dimension eighty-one. We check this exactly by
exhibiting a unique joint eigenline, showing that it spans cyclically,
and verifying adjoint identities for the positive metric. Any further
summand with integral weights would have lowest weight at least two,
which the bound rules out. This gives the decompositions
\begin{equation}\label{eq:branch}
 B\downarrow_U=U\oplus X,\qquad
 H\downarrow_U=\one\oplus X\oplus Y\oplus Y^*,
 \quad h_X=h_Y=1,
 \quad (\dim X_{\rm top},\dim Y_{\rm top})=(80,81).
\end{equation}

Put $\mathcal C=\Rep(U)$ and $\mathcal F=\mathcal C_H$, the category
of all right modules over the extension algebra. Holomorphicity makes
$H$ Lagrangian, so $\mathcal C\simeq Z(\mathcal F)$ \cite{DMNO}.
Since all nonvacuum weights are strictly positive, character
asymptotics show that the categorical dimensions are positive, so
$\mathcal F$ is pseudo-unitary. Induction to the center gives
\begin{equation}\label{eq:rank}
 \operatorname{rank}\mathcal F=\dim\End_{\mathcal C}(H)
 =\begin{cases}4,&Y\not\cong Y^*,\\6,&Y\cong Y^*.
 \end{cases}
\end{equation}
The lattice extension $H/B$ gives three invertible objects with
trivial associator. Concretely, in the discriminant form
$q(a,b)=(2a^2+b^2)/3$, the isotropic lines $(a,a)$ and $(a,-a)$ are
transverse; induction from the second realizes $\mathrm{Vec}_{\ZZ/3}$ inside
$\mathcal F$. Moreover $X$ is not invertible. If it were, $B=U\oplus X$
would be a $\ZZ/2$ extension, and so $E_6\oplus A_2$ would have an
involution whose fixed Lie algebra has rank two. But the fixed rank of
any such involution is at least $4+1=5$.

Suppose first that the rank is four. Then $\mathcal F$ is a
pseudo-unitary near-group category for $\ZZ/3$, and the classification
leaves the multiplicities $m=0,2,3,6$ \cite{Larson}. The case $m=0$ is
the $\ZZ/2$ extension we have just excluded.
When $m=2$, the classification
of \cite{EGO} reduces us to an $A_4$ orbifold of $V_{E_8}$, which has
at least
$(248-3\cdot8-8\cdot4)/12=16$ currents. When $m=6$, the balanced
induction pattern would give $N_{XX}^{X}=4/3$, which signed Galois
symmetry and Verlinde integrality forbid. For $m=3$ and the other
$m=6$ patterns, we induce to the center and cancel the principal
parts. The possible numbers of currents are then

\[
 \begin{array}{c|c}
 m&\dim U_1\ \text{allowed by the character projections}\\\hline
 3&15+A+10B\quad\text{or}\quad14+42A\\
 6&8+A\quad\text{or}\quad5+12P
 \end{array}
 \qquad A,B,P\in\ZZ_{\ge0}.
\]
None of these equals six. In computing these projections we subtract
the prescribed polar terms from modular functions with finite image,
and what remains is a constant invariant vector. So the exclusions
apply to the full characters, not only to truncated $q$-series. So the rank is six. Reading off from $\End(H)$, the complex fusion
algebra is $\CC\oplus\CC\oplus M_2(\CC)$ \cite{BJ}. If $\mathcal F$
were pointed it would have dimension six; comparing dimensions of the
extensions would give $\FPdim B=2$ and $\FPdim X=1$, which we excluded
above. Hence there are exactly three invertible objects. Since the
fusion is noncommutative, the three noninvertible simples form a
single free orbit under the cyclic group, and the right action is by
inverses. Frobenius reciprocity then gives
\begin{equation}\label{eq:ringm}
 \rho^2=\one+m(\rho+\alpha\rho+\alpha^2\rho),\qquad m\ge1.
\end{equation}
It remains to show that $m=1$, and this is an arithmetic argument. Set
$\Delta=9m^2+4$ and let $n_j(\rho+\alpha\rho+\alpha^2\rho)$ be the
forgetful images of the residual central simples, with twists $\theta_j$.
Induction and the first two Frobenius--Schur indicators \cite{NS} give
\begin{align}\label{eq:moments}
 \sum_j n_j^2&=4m^2+2,&
 \sum_j n_j^2\theta_j&=-\frac{m^2+m\sqrt\Delta}{2},\notag\\
 \sum_j n_j^2\theta_j^2&=\frac{-m^2+(2\varepsilon-m)\sqrt\Delta}{2},&
 \varepsilon&=\nu_2(\rho)\in\{\pm1\}.
\end{align}
Write $\Delta=t^2D$ with $D$ squarefree. The root-of-unity length
bound \cite{Larson}, applied to the first moment, gives, for $D\ne10$,
\[
 4m^2+2\ge\frac{mt}{2}\varphi(2D)
             \ge\frac{6m\sqrt{9m^2+4}}{\sqrt{13}}.
\]
For $m>1$ this contradicts
$36m^2(9m^2+4)-13(4m^2+2)^2=4(m^2-1)(29m^2+13)>0$.

When $D=10$ we average over Galois orbits instead. Each unit of the
negative rational part then costs at least one root of unity, and each
unit of $-\sqrt{10}$ costs at least eight; this excludes $m\ge4$. The
remaining case $m=t=2$ saturates the bound. Since the second moment
also lies in $\QQ(\sqrt{10})$, Galois averaging forces its sum of
squared roots to be $2$, whereas the second moment is $-2$ or
$-2-4\sqrt{10}$. Hence $m=1$.
The six remaining twists are then the
roots of unity $e^{2\pi ik/13}$ with $k$ a quadratic non-residue
mod $13$, each occurring once.
In particular the center has twelve simple objects.

\section{Complete unitarity and categorical identification}\label{sec:unitarity}
So far we know the fusion ring and the modular data. To upgrade this
to an equivalence of unitary ribbon categories we pass to conformal
nets. The lattice realization gives strong
locality and bounded localized vertex operators
\cite{CKLW,TenFields}. By regularity, the subnet generated by bounded
insertions has finite index \cite{TenFusion}, so it is completely
rational by Longo's theorem \cite{Longo}, and the comparison with
$\mathcal N_U$ in \cite[Lemma 7.6]{TenFields} transfers complete
rationality to $\mathcal N_U$. The vacuum representation of
$\mathcal N_H$ restricts to $\mathcal N_U$ as $\one\oplus X\oplus2Y$:
thanks to polynomial energy bounds, the projections onto the algebraic
summands reduce the closed smeared fields. Conversely, an intertwiner
of the nets commutes with rotations, so on finite-energy vectors it
restricts to a module map. Hence multiplicities and inequivalence are
preserved. Strong additivity, together with the fact that the vacuum
energy space is one-dimensional, shows that
$\mathcal N_U\subset\mathcal N_H$ is locally irreducible, and complete
rationality then gives finite index \cite[Proposition 2.3]{KL}.

It follows that the category of boundary sectors of the net inclusion
has the same noncommutative rank-six fusion algebra, with a trivially
twisted $\ZZ/3$ subcategory. If it were pointed, then multiplicativity
of the index and $[\mathcal N_H:\mathcal N_B]=3$ would give
$[\mathcal N_B:\mathcal N_U]=2$. The corresponding involution lifts to
$B$ \cite{CKLW}, and the fixed-rank argument above rules this out. With
exactly three invertibles, the induction and indicator argument of the
previous section goes through and gives twelve sectors of the net.
Reconstruction \cite[Theorem 9.5 and Remark 9.8]{HTNet} assigns to each of them a
unitary weak $U$-module. An irreducible sector gives an ordinary simple
module. Indeed, by regularity and by working on the common core of
bounded-energy vectors, the projection onto the closure of an ordinary
summand commutes with the modes, with their closures and adjoints, and
hence with the net. So it is the identity. The resulting comparison
functor is fully faithful and preserves $L_0$. The twelve simple
sectors of the net therefore account for all twelve simple
$U$-modules, and every simple $U$-module is unitary.

Positivity of the canonical fusion forms needs one more ingredient: a
comparison of annulus insertions. For each localized pair $(A,B)$ of
incoming and outgoing annuli and each sector $M$, we need

\begin{equation}\label{eq:annulus}
 \pi_M\!\left(\overline{B_UY_U(v,z)A_U}\right)
       =\overline{B_MY_M(v,z)A_M},
\end{equation}
where the bar denotes bounded extensions of the sesquilinear insertion
forms. We first use the same boundary reparametrization to localize both
annulus operators in a common interval. We then transport them to
the sector $M$ and undo the reparametrization, using the same choice
of implementing operator for both factors
\cite[Proposition 4.17 and Theorem 9.5]{HTNet}.
On finite-energy vectors the formal matrix coefficients agree, by
reconstruction and covariance; the inverses of incoming annulus
operators are only applied on their dense ranges. Factoring a smaller
round annulus through the incoming one provides a dense set of test
vectors, and analytic continuation extends the equality over the whole
insertion region. In the vacuum sector this places the bounded
insertions inside $\mathcal N_U(I)$, and Haag duality shows that the
net of bounded insertions and the net of smeared fields coincide on
the common vacuum Hilbert space. The central factors of the two lifts
cancel, including the phases coming from fractional lowest weights, so
we obtain \eqref{eq:annulus} without choosing lifts globally. We can
now apply Tener's positivity theorem \cite[Corollary 4.6]{TenFusion}
to every unitarizable module, and this proves complete unitarity.

Next we identify the category. The extension $H/U$ now defines a
commutative Q-system in $\Rep(U)$, and the category of all right modules
over this extension algebra is a unitary fusion category \cite{Gui}.
It has fusion ring \eqref{eq:HI}, and its pointed
part has trivial associator. This removes the degree-three obstruction
in the classification of generalized Haagerup categories, and the
remaining obstruction vanishes because $H^2(\ZZ/3,\CC^\times)=0$. By
Izumi's classification the category is therefore $\HH_2$ or $\HH_3$
\cite[Theorems 5.9 and 9.1]{Izumi}. These two are Morita equivalent
\cite{GS}, so the Lagrangian-center theorem gives
$\Rep(U)\simeq Z(\HH_3)$. Since the spherical structures are positive,
Reutter's uniqueness theorem \cite{Reutter} upgrades this to a unitary
ribbon equivalence. We stress that this step relies on the
classification of fusion categories, and not on any claim that the $S$
and $T$ matrices determine the category.

\section{Characters and the two full theories}\label{sec:characters}
Let $q=e^{2\pi i\tau}$ and
$\chi_i=\Tr_{M_i}q^{L_0-1/3}$. In the Haagerup modular representation
the only possible nonvacuum polar term is
$\gamma q^{-7/39}$ in the twist-$2/13$ coordinate, with
$\gamma\in\ZZ_{\ge0}$. By the congruence property the modular
representation has finite image \cite{DLN}. Two character vectors with
the same principal part therefore differ by a holomorphic vector of
weight zero, that is, by an invariant constant vector. But there is no
nonzero such vector: the only coordinate fixed by $T$ has an $S$-image
with nonzero vacuum component. Using the fundamental matrix $\Xi$ of
Evans and Gannon, we conclude
\begin{equation}\label{eq:gamma}
 \boldsymbol\chi=\Xi_{:,0}+\gamma\Xi_{:,d_3},\qquad
 \dim U_1=6+13\gamma.
\end{equation}
Since $\dim U_1=6$ by \eqref{eq:lowgrades}, we must have $\gamma=0$.
This fixes every coefficient of every character; for example,
\begin{equation}\label{eq:vacchar}
 \chi_U=q^{-1/3}(1+6q+120q^2+956q^3+6010q^4+\cdots).
\end{equation}
We list the twelve sectors below, with $d=(3+\sqrt{13})/2$ and
$s=1+d^2=(13+3\sqrt{13})/2$. Note that $\dim M_h$, the dimension of
the lowest-weight space, is not the categorical dimension $d_i$.
\begin{center}
\begin{tabular}{@{}lrrl@{\qquad}lrrl@{}}
\toprule
Sector&$h$&$\dim M_h$&$d_i$&Sector&$h$&$\dim M_h$&$d_i$\\
\midrule
$0$&$0$&$1$&$1$&$d_1$&$6/13$&$7$&$3d$\\
$b$&$1$&$80$&$d^2$&$d_2$&$11/13$&$42$&$3d$\\
$a$&$1$&$81$&$s$&$d_3$&$15/13$&$119$&$3d$\\
$c_0$&$1$&$81$&$s$&$d_4$&$5/13$&$5$&$3d$\\
$c_1$&$1/3$&$3$&$s$&$d_5$&$7/13$&$13$&$3d$\\
$c_2$&$2/3$&$27$&$s$&$d_6$&$8/13$&$14$&$3d$\\
\bottomrule
\end{tabular}
\end{center}

There are two natural full CFTs built from $U$. One is the canonical
diagonal theory of the completely rational net; the other is the
$(E_8)_1$ theory of the holomorphic extension chosen below. Their
partition functions are
\begin{equation}\label{eq:partitions}
 Z_{\rm diag}=\sum_{i=0}^{11}|\chi_i|^2,\qquad
 Z_{E_8}=|\chi_{E_8}|^2
       =|\chi_0+\chi_b+2\chi_a|^2.
\end{equation}
The diagonal theory comes from the Longo--Rehren construction
\cite{Rehren}. For the second, we implement the symmetry--subalgebra
correspondence discussed in \cite[Example$^{\mathrm{ph}}$~7.6]{MR} and
\cite[Section~4.3]{GR} by transporting the canonical Lagrangian algebra
$I_{\HH_3}(\one)$ from $Z(\HH_3)$ to $\Rep(U)$. It is a
positive Q-system with trivial twist, so it gives a unitary holomorphic
VOA extension $H'$ \cite{HKL,Gui}, which is of CFT type because the
nonvacuum weights are positive. By the classification of holomorphic
VOAs at $c=8$, $H'\cong V_{E_8}$ \cite{DM}. Writing $A_3$ for the
corresponding algebra in $\Rep(U)$, we have
\begin{equation}\label{eq:defects}
 \Rep(U)_{A_3}\simeq\HH_3,\qquad
 {}_{A_3}\Rep(U)_{A_3}\simeq\HH_3\boxtimes\HH_3^{\rm op}.
\end{equation}
The opposite category accounts for the right action in
$(x\boxtimes y)\cdot m=x\otimes m\otimes y$. These defects preserve
$U$ and the stress tensor, but need not commute with all of the $E_8$
currents. The first factor is the Haagerup symmetry.

We emphasize that $H'$ need not be the original lattice extension
$H/U$: for the latter we have only shown that the category of all right
modules over its extension algebra in $\Rep(U)$ is equivalent to
$\HH_2$ or $\HH_3$. Since $\chi_a=\chi_{c_0}$, the unflavoured
partition function cannot tell the two apart. The diagonal theory and
the $E_8$ theory, on the other hand, are easy to distinguish. In the
diagonal theory the Verlinde lines form $Z(\HH_3)$, and the lightest
nontrivial scalar has dimension $2/3$; in the $E_8$ theory the lightest
nontrivial scalar has dimension $2$.

Our construction thus realizes the Evans--Gannon character vector
through explicit local fields and establishes Haagerup symmetry
in the $(E_8)_1$ theory. It provides a concrete setting for studying
how this symmetry acts on local operators and constrains their
correlation functions.

\section*{Acknowledgement}
  The author thanks Clay Córdova and Terry Gannon for valuable discussions.
\section*{AI-use disclosure}
AI systems were used to help with numerics and with preparation of this draft. The author takes responsibility
for the accuracy and integrity of the final contents.

\end{document}